\documentclass[10pt,a4paper]{article}
\usepackage{sita2011}
\usepackage{amsmath,amssymb,amsthm}
\usepackage{mathptmx}
\usepackage{cite}
\usepackage{graphicx}
\usepackage{psfrag}
\usepackage{stfloats}
\DeclareMathAlphabet{\mathcal}{OMS}{cmsy}{m}{n}
\SetSymbolFont{operators}{bold}{OT1}{txr}{bx}{n}
\SetSymbolFont{letters}{bold}{OML}{txmi}{bx}{it}
\SetSymbolFont{symbols}{bold}{OMS}{txsy}{bx}{n}
\SetSymbolFont{largesymbols}{bold}{OMX}{txex}{bx}{n}
\usepackage{bm}
\title{Statistical Mechanical Analysis of Low-Density Parity-Check Codes on General Markov Channel}
\author{
Ryuhei~Mori 
\thanks{The work of R. Mori was supported by the Grant-in-Aid for Scientific Research for JSPS Fellows (22$\cdot$5936).}
\thanks{
Department of Systems Science, Graduate School of Informatics, Kyoto University,
Kyoto, 606-8501, Japan
(email: rmori@sys.i.kyoto-u.ac.jp, tt@i.kyoto-u.ac.jp).}%
\and 
Toshiyuki~Tanaka
\samethanks{2}}

\allowdisplaybreaks[1]

\newtheorem{theorem}{Theorem}
\newtheorem{lemma}[theorem]{Lemma}

\def\extr{\mathop{\rm extr}}

\abstract{
Low-density parity-check (LDPC) codes on symmetric memoryless channels
have been analyzed using statistical physics by several authors.
In this paper, statistical mechanical analysis of LDPC codes is performed for asymmetric memoryless channels and general Markov
channels. It is shown that the saddle point equations of the replica symmetric solution for a Markov channel is equivalent to the density
evolution of the belief propagation on the factor graph representing LDPC codes on the Markov channel.
The derivation uses the method of types for Markov chain.
}


\begin{document}
\maketitle

\section{Introduction}
The replica method is a tool for the evaluation of free energy in statistical physics.
Although the mathematical rigorousness of the replica method has not been proved, the replica method has been used not only for
problems of statistical physics but for problems in several areas including information theory~\cite{tanaka2002statistical}.
The analysis of low-density parity-check (LDPC) codes using the replica method is shown for memoryless symmetric channels
by Murayama et al.~\cite{murayama2000statistical} and Montanari~\cite{montanari2001glassy}.
After some period, LDPC codes on binary-output asymmetric channel is analyzed by using the replica method by Neri et al.~\cite{neri2008gallager}.
On the other hand, LDPC codes on channel with memory is also of great interest in wireless communication and magnetic recording~\cite{eckford2004thesis}.
LDPC codes on restricted types of finite-state Markov channel are also analyzed by Neri and Skantzos~\cite{neri2009statistical}.
In this paper, we analyze LDPC codes on general asymmetric memoryless channel and general Markov channel by using the replica method.
In the analysis, we use the method of types for Markov chain
in order to generalize the method recently obtained in~\cite{mori2011connection}.

Despite lack of mathematical rigour, the replica method is known to yield rigorous bounds
in some cases~\cite{franz2003non-poisson}, \cite{montanari2005tight}.
Moreover, a tight bound of MAP threshold, which is derived in an elegant way in~\cite{measson2009generalized},
is found to empirically coincide with the result of the replica method.

\section{Memoryless channel}
\subsection{LDPC codes and the number of codewords}
In this paper, we deal with regular LDPC codes.
Generalization for irregular LDPC codes is straightforward as in~\cite{mori2011connection}.
We consider factor graph (Tanner graph) including $N$ variable nodes and $N(1-R)$ factor nodes (check nodes).
The degrees of variable nodes and factor nodes are $l$ and $r$, respectively, 
for which the relation $R=1-l/r$ holds.
Edge connection is randomly chosen uniformly from $(Nl)!$ possible connections.
Let $\mathcal{X}=\{0,1\}$ and $f(\bm{x}):\mathcal{X}^r \to \{0,1\}$ be a factor function taking 1 when the number of 1s in $\bm{x}$ is even,
and taking 0 otherwise.
Let $\bm{x}_{\partial a}$ be the neighborhood of a factor node $a$.
The uniform distribution on codewords $\bm{x}\in\mathcal{X}^N$ of the LDPC code 
defined by the Tanner graph is given by 
\begin{align*}
p(\bm{x}) &:= \frac1{Z_0}\prod_{a} f(\bm{x}_{\partial a})\\
Z_0 &:= \sum_{\bm{x}\in\mathcal{X}^N} \prod_{a}f(\bm{x}_{\partial a}).
\end{align*}
Here, $Z_0$ is a constant for the normalization,
which in this case is the number of codewords in an LDPC code.
Let $\mathbb{E}[\cdot]$ be the expectation with respect to the edge connections.
In~\cite{mori2011connection}, the following lemma is shown for general $f(\bm{x})$.
\begin{lemma}
\begin{multline}\label{eq:rate}
\lim_{N\to\infty}\frac1N\log\mathbb{E}[Z_0]\\
=\max_{(m_\mathrm{f\to v},m_\mathrm{v\to f})\in\mathcal{R}}
\left\{ \log Z_\mathrm{v} + \frac{l}{r} \log Z_\mathrm{f} 
-l\log Z_\mathrm{e}\right\}
\end{multline}
where $\mathcal{R}$ denotes the set of saddle points of the function 
for which the maximum is taken, and where
\begin{align*}
Z_\mathrm{v} &:= \sum_{x\in\mathcal{X}} m_\mathrm{f\to v}(x)^l,&
Z_\mathrm{f} &:= \sum_{\bm{x}\in\mathcal{X}^r} f(\bm{x}) \prod_{i=1}^r m_\mathrm{v\to f}(x_i)\\
Z_\mathrm{e} &:= \sum_{x\in\mathcal{X}} m_\mathrm{f\to v}(x)m_\mathrm{v\to f}(x).
\end{align*}
The saddle point equations determining the two probability functions 
$m_\mathrm{v\to f}(x),\,m_\mathrm{f\to v}(x)$, both defined on $\mathcal{X}$, are
\begin{equation}
\begin{split}
m_\mathrm{v\to f}(x) &\propto m_\mathrm{f\to v}(x)^{l-1}\\
m_\mathrm{f\to v}(x) &\propto \frac1r\sum_{i=1}^r \sum_{\bm{x}\in\mathcal{X}^r, x_i=x} f(\bm{x})\prod_{j=1, j\ne i}^rm_\mathrm{v\to f}(x_j).
\end{split}
\label{eq:saddle}
\end{equation}
\end{lemma}
In the case of regular LDPC codes, it is proved~\cite{RiU05/LTHC} that 
the maximizer for the right-hand side of~\eqref{eq:rate} is given by $m_\mathrm{v\to f}(x)=m_\mathrm{f\to v}(x)=1/2$, 
for which
\begin{equation}
\lim_{N\to\infty} \frac1N \log \mathbb{E}[Z_0] = R.
\label{eq:asump}
\end{equation}
Since each independent factor reduces the number of codewords by half, 
the inequality $(1/N)\log Z_0 \ge R$ always holds.
Hence, for regular LDPC codes,
\begin{equation*}
\lim_{N\to\infty} \frac1N \mathbb{E}[\log Z_0] = R
\end{equation*}
and for any $\epsilon>0$,
\begin{equation}\label{eq:concentration}
\frac1N \log Z_0 \le R+\epsilon, \hspace{2em} \text{with probability } 1- o(1).
\end{equation}
The result of this paper can be generalized to any irregular LDPC code ensemble satisfying~\eqref{eq:asump}.

\subsection{Replica analysis}
In this subsection, we consider LDPC codes transmitted on an asymmetric memoryless channel.
Let $\mathcal{Y}$ be the output alphabet and $W(y\mid x)$ be a probability 
of an output $y\in\mathcal{Y}$ when $x\in\mathcal{X}$ is transmitted.
The a posteriori distribution of a codeword $\bm{x}\in\mathcal{X}^N$ given an output $\bm{y}\in\mathcal{Y}^N$ is
\begin{align*}
p(\bm{x}\mid \bm{y}) &:= \frac1Z\prod_{a} f(\bm{x}_{\partial a}) \prod_{i=1}^{N} W(y_i \mid x_i)\\
Z &:= \sum_{\bm{x}\in\mathcal{X}^N} \prod_{a}f(\bm{x}_{\partial a})\prod_{i=1}^{N} W(y_i \mid x_i).
\end{align*}
Let $\mathbf{E}[\cdot]$ be the expectation with respect to the output $\bm{y}\in\mathcal{Y}^N$ of the channel, which obeys
\begin{equation*}
p(\bm{y}) := \frac1{Z_0}\sum_{\bm{x}\in\mathcal{X}^N}\prod_{a}f(\bm{x}_{\partial a})\prod_{i=1}^{N} W(y_i \mid x_i).
\end{equation*}
The aim of this paper is to evaluate
\begin{equation}\label{eq:FE}
\lim_{N\to\infty}(1/N)\mathbb{E}[\mathbf{E}[\log Z]]
\end{equation}
which gives the expected conditional entropy of codeword given an output~\cite{macris2007griffith}.
The quantity~\eqref{eq:FE} is evaluated via the replica method. 
First, from~\eqref{eq:concentration}, $Z_0$ in the definition of $p(\bm{y})$ can be replaced by $2^{NR}$.
Then, one obtains
\begin{align*}
Z^n &= \sum_{\bm{x}\in(\mathcal{X}^n)^N} \left(\prod_{a}\prod_{k=1}^nf(\bm{x}_{\partial a}^{(k)})\right)\prod_{i=1}^{N}\prod_{k=1}^n W(y_i \mid x_i^{(k)})\\
\mathbf{E}[Z^n] &= \sum_{\bm{y}\in\mathcal{Y}^N} p(\bm{y})
\sum_{\bm{x}\in(\mathcal{X}^n)^N} \left(\prod_{a}\prod_{k=1}^nf(\bm{x}_{\partial a}^{(k)})\right)\prod_{i=1}^{N}\prod_{k=1}^n W(y_i \mid x_i^{(k)})\\
&=
\frac1{2^{NR}}
\sum_{\bm{x}\in(\mathcal{X}^{(n+1)})^N} \prod_{a}\prod_{k=0}^{n}f(\bm{x}_{\partial a}^{(k)})\\
&\quad\cdot\prod_{i=1}^{N}\left(\sum_{y\in\mathcal{Y}}\prod_{k=0}^{n} W(y \mid x_i^{(k)})\right)
\end{align*}
where $x^{(k)}\in\mathcal{X}$ and $\bm{z}^{(k)}\in\mathcal{X}^r$ denote the variables corresponding to
$k$th replica of $\bm{x}\in\mathcal{X}^{n+1}$ and $\bm{z}\in(\mathcal{X}^{n+1})^r$, respectively.
By using the method in~\cite{mori2011connection},
for $n\in\{0,1,2,\dotsc\}$, one obtains
\begin{multline}\label{eq:nth}
\lim_{N\to\infty}\frac1N\log \mathbb{E}[\mathbf{E}[Z^n]]\\
 = \max_{(m_\mathrm{f\to v}(\bm{x}),m_\mathrm{v\to f}(\bm{x}))\in\mathcal{R}}\left\{\log Z_\mathrm{v} + \frac{l}{r}\log Z_\mathrm{f} - l\log Z_\mathrm{e}\right\} -R
\end{multline}
where $\mathcal{R}$ denotes the set of saddle points of the function for which the maximum is taken, and where
\begin{align*}
Z_\mathrm{v} &:= \sum_{\bm{x}\in\mathcal{X}^{n+1}} \left(\sum_{y\in\mathcal{Y}}\prod_{k=0}^{n} W(y \mid x_i^{(k)})\right)m_\mathrm{f\to v}(\bm{x})^l\\
Z_\mathrm{f} &:= \sum_{\bm{x}\in(\mathcal{X}^{n+1})^r} \prod_{k=0}^{n}f(\bm{x}^{(k)}) \prod_{i=1}^r m_\mathrm{v\to f}(\bm{x}_i)\\
Z_\mathrm{e} &:= \sum_{\bm{x}\in\mathcal{X}^{n+1}} m_\mathrm{f\to v}(\bm{x})m_\mathrm{v\to f}(\bm{x}).
\end{align*}
 The replica symmetry (RS) assumption that we use is
\begin{equation}
\begin{split}
m_\mathrm{v\to f}(\bm{x}) &= m_\mathrm{v\to f}(x_0)\int \prod_{i=1}^nM_\mathrm{v\to f}(x_i)\mathrm{d}\Phi(M_\mathrm{v\to f}\mid x_0)\\
m_\mathrm{f\to v}(\bm{x}) &= m_\mathrm{f\to v}(x_0)\int \prod_{i=1}^nM_\mathrm{f\to v}(x_i)\mathrm{d}\hat{\Phi}(M_\mathrm{f\to v}\mid x_0)
\end{split}
\label{eq:RS}
\end{equation}
where $\Phi$ and $\hat{\Phi}$ are distributions of measures on $\mathcal{X}$.
In~\eqref{eq:RS}, the notation $x_i$ is used instead of $x^{(i)}$ for $i\in\{0,\dotsc,n\}$.
From~\eqref{eq:asump}, $R$ in the right-hand side of~\eqref{eq:nth} can be replaced by the right-hand side of~\eqref{eq:rate}.
On the RS assumption~\eqref{eq:RS}, we obtain the following lemma.
\begin{lemma}
For any $n\in\mathbb{R}$,
\begin{multline*}
\lim_{N\to\infty}\frac1N\log \mathbb{E}[\mathbf{E}[Z^n]]\\
 = \extr_{(\Phi,\hat{\Phi},m_\mathrm{f\to v},m_\mathrm{v\to f})}
\left\{\log Z_\mathrm{v}(n) + \frac{l}{r}\log Z_\mathrm{f}(n) - l\log Z_\mathrm{e}(n)\right\}
\end{multline*}
where $\extr_x\{F(x)\}$ denotes $F(x^*)$ where $x^*$ is the saddle point of $F(x)$, and
where
\begin{align*}
Z_\mathrm{v}(n) &:= \sum_{x\in\mathcal{X}}\frac{m_\mathrm{f\to v}(x)^l}{\sum_{x\in\mathcal{X}} m_\mathrm{f\to v}(x)^l} \sum_{y\in\mathcal{Y}} W(y \mid x)\\
&\quad\cdot\int\prod_{i=1}^l\mathrm{d}\hat{\Phi}(M_\mathrm{f\to v}^{(i)}\mid x)
\left(\sum_{x\in\mathcal{X}} W(y \mid x)\prod_{i=1}^lM_\mathrm{f\to v}^{(i)}(x)\right)^n\\
Z_\mathrm{f}(n) &:=  
\sum_{\bm{x}\in\mathcal{X}^r}\frac{ f(\bm{x}) \prod_{i=1}^rm_\mathrm{v\to f}(x_i)}
{\sum_{\bm{x}\in\mathcal{X}^r} f(\bm{x}) \prod_{i=1}^rm_\mathrm{v\to f}(x_i)}\\
&\quad\cdot\int\prod_{i=1}^r\mathrm{d}\Phi(M_\mathrm{v\to f}^{(i)}\mid x_i)
\left(\sum_{\bm{x}\in\mathcal{X}^r} f(\bm{x}) \prod_{i=1}^r M_\mathrm{v\to f}^{(i)}(x_i)\right)^n\\
Z_\mathrm{e}(n) &:= 
\sum_{x\in\mathcal{X}} \frac{m_\mathrm{f\to v}(x) m_\mathrm{v\to f}(x)}{\sum_{x\in\mathcal{X}}m_\mathrm{f\to v}(x) m_\mathrm{v\to f}(x)}
\int \mathrm{d}\Phi(M_\mathrm{v\to f}\mid x)\\
&\quad\cdot \mathrm{d}\hat{\Phi}(M_\mathrm{f\to v}\mid x)
\left(\sum_{x\in\mathcal{X}} M_\mathrm{f\to v}(x)M_\mathrm{v\to f}(x)\right)^n.
\end{align*}
\end{lemma}
From this lemma, the RS solution is obtained straightforwardly~\cite{mori2011connection}.
\begin{theorem}\label{thm:asym}
\begin{align*}
&\lim_{N\to\infty}\frac1N\mathbb{E}[\mathbf{E}[\log Z]]=\lim_{n\to 0}\frac1n\lim_{N\to\infty}\frac1N\log\mathbb{E}[\mathbf{E}[Z^n]]\\
& = \extr_{(\Phi,\hat{\Phi},m_\mathrm{v\to f},m_\mathrm{f\to v})}\Bigg\{
\frac{l}{r} \sum_{\bm{x}\in\mathcal{X}^r}\frac{f(\bm{x}) \prod_{i=1}^rm_\mathrm{v\to f}(x_i)}
{\sum_{\bm{x}\in\mathcal{X}^r} f(\bm{x}) \prod_{i=1}^rm_\mathrm{v\to f}(x_i)}\\
&\quad\cdot \int\prod_{i=1}^r\mathrm{d}\Phi(M_\mathrm{v\to f}^{(i)}\mid x_i)
\log \left(\sum_{\bm{x}\in\mathcal{X}^r} f(\bm{x}) \prod_{i=1}^r M_\mathrm{v\to f}^{(i)}(x_i)\right)\\
& + \sum_{x\in\mathcal{X}}\frac{m_\mathrm{f\to v}(x)^l}
{\sum_{x\in\mathcal{X}}m_\mathrm{f\to v}(x)^l}\sum_{y\in\mathcal{Y}}W(y \mid x)\int\prod_{i=1}^l\mathrm{d}\hat{\Phi}(M_\mathrm{f\to v}^{(i)}\mid x)\\
&\quad\cdot\log \left(\sum_{x\in\mathcal{X}} W(y \mid x)\prod_{i=1}^lM_\mathrm{f\to v}^{(i)}(x)\right)\\
& - l\sum_{x\in\mathcal{X}}\frac{m_\mathrm{f\to v}(x)m_\mathrm{v\to f}(x)}{\sum_{x\in\mathcal{X}}m_\mathrm{f\to v}(x)m_\mathrm{v\to f}(x)}
\int\mathrm{d}\Phi(M_\mathrm{v\to f}\mid x)\\
&\quad\cdot \mathrm{d}\hat{\Phi}(M_\mathrm{f\to v}\mid x)
\log \left(\sum_{x\in\mathcal{X}} M_\mathrm{f\to v}(x)M_\mathrm{v\to f}(x)\right)\Bigg\}
\end{align*}
where the saddle point equations are~\eqref{eq:saddle} and
\begin{align*}
&m_\mathrm{f\to v}(x)\hat{\Phi}(M_\mathrm{f\to v}\mid x)\\
&= \frac1r \sum_{k=1}^r 
\sum_{\bm{x}\in\mathcal{X}^r, x_k=x}
\frac{f(\bm{x}) \prod_{i=1,i\ne k}^r m_\mathrm{v\to f}(x_i)}{\sum_{\bm{x}\in\mathcal{X}^r} f(\bm{x}) \prod_{i=1,i\ne k}^rm_\mathrm{v\to f}(x_i)}\\
&\quad\cdot \int\prod_{i=1, i\ne k}^r \mathrm{d}\Phi(M_\mathrm{v\to f}^{(i)}\mid x_i)\\
&\quad\cdot \delta\left(M_\mathrm{f\to v}, \frac{\sum_{\bm{x}\in\mathcal{X}^r\setminus x_k}f(\bm{x})\prod_{i\ne k}M_\mathrm{v\to f}^{(i)}(x_i)}
{\sum_{\bm{x}\in\mathcal{X}^r}f(\bm{x})\prod_{i\ne k}M_\mathrm{v\to f}^{(i)}(x_i)}\right)\\
&\Phi(M_\mathrm{v\to f}\mid x)
 = \sum_{y\in\mathcal{Y}}
W(y \mid x)\int\prod_{i=1}^{l-1}\mathrm{d}\hat{\Phi}(M_\mathrm{f\to v}^{(i)}\mid x)\\
& \quad\cdot \delta\left(M_\mathrm{v\to f}, \frac{W(y\mid x)\prod_{i=1}^{l-1}M_\mathrm{f\to v}^{(i)}(x)}{\sum_{x\in\mathcal{X}}W(y\mid x)\prod_{i=1}^{l-1}M_\mathrm{f\to v}^{(i)}(x)}\right).
\end{align*}
\end{theorem}
As mentioned above, all $m_\mathrm{f\to v}(x)$ and $m_\mathrm{v\to f}(x)$ in Theorem~\ref{thm:asym} can be replaced by 1/2.
The saddle point equations in Theorem~\ref{thm:asym} are equivalent to the density evolution of belief propagation for asymmetric memoryless channel~\cite{wang2005density}.
Although the RS assumption is sufficient for symmetric channels~\cite{montanari2001glassy},
for asymmetric channels, 1-step replica symmetry breaking (1RSB) should be considered generally~\cite{neri2008gallager}.
The equations for 1RSB solution are obtained straightforwardly from~\eqref{eq:nth} and the 1RSB assumption.
However, numerically solving the saddle point equations for 1RSB solution is much more elaborate.
For LDPC codes, there exists a \textit{frozen solution\/} which is the minimum of the RS solution
with respect to temperature~\cite{montanari2001glassy}, \cite{martin2005random}, \cite{neri2008gallager}.
While the frozen solution can be relatively easily obtained like the RS solution, it is not obvious whether the frozen solution is the correct solution~\cite{neri2008gallager}.

\section{Markov channel}
In this section, we consider LDPC codes on the general Markov channel.
The space of states is denoted by $\mathcal{S}$. 
The state-dependent channel is denoted by $W(y\mid x,s)$, and 
the transition probability of states by $V(s' \mid y, x, s)$, 
for $y\in\mathcal{Y}$, $x\in\mathcal{X}$, and $s,\, s'\in\mathcal{S}$.
The probability distribution of the initial state $s\in\mathcal{S}$ is denoted by $V_0(s)$.
When the transition probability is independent of $y$,
we call the channel an intersymbol-interference channel.
When the transition probability is independent of $y$ and $x$,
we call the channel a finite-state Markov channel.
The a posteriori probability of a codeword $\bm{x}\in\mathcal{X}^N$ given an output $\bm{y}\in\mathcal{Y}^N$ is
\begin{align*}
p(\bm{x}\mid \bm{y}) &:= \frac1Z\sum_{\bm{s}\in\mathcal{S}^N}\prod_{a} f(\bm{x}_{\partial a})
\prod_{i=1}^{N} W(y_i \mid x_i, s_i)\\
&\quad\cdot V_0(s_1)\prod_{i=1}^{N-1}V(s_{i+1}\mid y_i, x_i, s_i)\\
Z &:= \sum_{\bm{x}\in\mathcal{X}^N}\sum_{\bm{s}\in\mathcal{S}^N} \prod_{a}f(\bm{x}_{\partial a})\prod_{i=1}^{N} W(y_i \mid x_i, s_i)\\
&\quad\cdot V_0(s_1)\prod_{i=1}^{N-1}V(s_{i+1}\mid y_i, x_i, s_i).
\end{align*}
Similarly to the previous subsection, let $\mathbf{E}[\cdot]$ denote the expectation with respect to $\bm{y}\in\mathcal{Y}^N$ obeying
\begin{align*}
p(\bm{y}) &:= \frac1{Z_0}\sum_{\bm{x}\in\mathcal{X}^N}
\sum_{\bm{s}\in\mathcal{S}^N}\prod_{a} f(\bm{x}_{\partial a}) \prod_{i=1}^{N} W(y_i \mid x_i, s_i)\\
&\quad\cdot V_0(s_1)\prod_{i=1}^{N-1}V(s_{i+1}\mid y_i, x_i, s_i).
\end{align*}
As in the previous subsection, 
we consider~\eqref{eq:FE} with 
$Z_0$ replaced by $2^{NR}$.
Then, one obtains
\begin{align*}
\mathbf{E}[Z^n] &= \frac1{2^{NR}}\sum_{\bm{x}\in(\mathcal{X}^{n+1})^N}\sum_{\bm{s}\in(\mathcal{S}^{n+1})^N}
\left(\prod_{a}\prod_{k=0}^nf(\bm{x}_{\partial a}^{(k)})\right)\prod_{k=0}^nV_0(s_1^{(k)})\\
&\quad\cdot\prod_{i=1}^{N}\left(\sum_{y\in\mathcal{Y}}\prod_{k=0}^n W(y \mid x_{i-1}^{(k)}, s_{i-1}^{(k)}) V(s_{i}^{(k)}\mid y, x_{i-1}^{(k)}, s_{i-1}^{(k)})\right).
\end{align*}
Let $\mathcal{T}_n\subseteq \mathcal{S}^{n+1}$ be the largest set such that for any $\bm{s}_2\in\mathcal{T}_n$ 
\begin{equation*}
\sum_{\bm{s}_1\in\mathcal{T}_n}\sum_{y\in\mathcal{Y}}\prod_{k=0}^n\left(\sum_{x_2\in\mathcal{X}} W(y\mid x_2,s_2^{(k)})V(s_1^{(k)}\mid y,x_2,s_2^{(k)})\right) > 0.
\end{equation*}
In order to use Sanov's theorem~\cite{dembo2009large} 
from
the method of types for Markov chain~\cite{whittle1955some,billingsley1961statistical,csiszar1987conditional},
we assume that the Markov chain on $\mathcal{T}_n$ defined by
the transition probabilities
\begin{equation*}
Q_n(\bm{s}_1\mid \bm{s}_2) :\propto \sum_{y\in\mathcal{Y}}\prod_{k=0}^n\left(\sum_{x_2\in\mathcal{X}}W(y\mid x_2,s_2^{(k)})V(s_1^{(k)}\mid y,x_2,s_2^{(k)})\right)
\end{equation*}
is irreducible for each $n\in\{0,1,2,\dotsc\}$.
Then, we obtain the following maximization problem like LDPC codes for memoryless channel~\cite{mori2011connection},
\begin{align}
&\lim_{N\to\infty}\frac1N\log \mathbb{E}[\mathbf{E}[Z^n]]=\sup\Bigg\{
H(X_1,S_1\mid X_2,S_2) - lH(X_1,S_1)\nonumber\\
&\quad  + \frac{l}{r} H(U_1,\dotsc,U_r, T_1,\dotsc,T_r)
+\frac{l}{r}\left\langle \log \prod_{k=0}^nf(\bm{U}^{(k)})\right\rangle\nonumber\\
&\quad +\left\langle\log\left(\sum_{y\in\mathcal{Y}}\prod_{k=0}^nW(y\mid X_2^{(k)},S_2^{(k)})V(S_1^{(k)}\mid y,X_2^{(k)},S_2^{(k)})\right)\right\rangle
\Bigg\}\nonumber\\
&\quad-R.
\label{eq:type}
\end{align}
Here, $X_1$ and $X_2$ are random variables on $\mathcal{X}^{n+1}$.
$S_1$ and $S_2$ are random variables on $\mathcal{T}_n$ or equivalently on $\mathcal{S}^{n+1}$.
$U_i$ and $T_i$ are random variables on $\mathcal{X}^{n+1}$ and $\mathcal{S}^{n+1}$, respectively, for $i\in\{1,\dotsc,r\}$.
The notation $\langle\cdot\rangle$ denotes the expectation with respect to the random variables.
The supremum is taken with respect to $(X_1,S_1,X_2,S_2)$ and $(U_1,\dotsc,U_r,T_1,\dotsc,T_r)$ on the following conditions
\begin{itemize}
\item $(X_1, S_1)$ and $(X_2, S_2)$ have the same distribution
\item $(X_1, S_1)$ and $(U_K, T_K)$ have the same distribution
\end{itemize}
where $K$ denotes the uniform random variable on a set $\{1,\dotsc,r\}$.
\begin{figure*}[!b]
\hrulefill
\setcounter{equation}{9}
\begin{align}
&\extr_{\Phi,\hat{\Phi},\Psi,\hat{\Psi}}\Bigg\{
\sum_{(x_1,x_2,s_2)\in\mathcal{X}\times\mathcal{X}\times\mathcal{S}}
m_\mathrm{L\to s}(s_2) \frac{m_\mathrm{f\to v}(x_1)^l}{\sum_{x_1\in\mathcal{X}}m_\mathrm{f\to v}(x_1)^l}
\frac{m_\mathrm{f\to v}(x_2)^l}{\sum_{x_2\in\mathcal{X}}m_\mathrm{f\to v}(x_2)^l}
\sum_{(y,s_1)\in\mathcal{Y}\times\mathcal{S}}
W(y\mid x_2,s_2)V(s_1\mid y,x_2,s_2) \nonumber\\
&\quad\cdot\int\mathrm{d}\Psi(M_\mathrm{R\to v}\mid x_1,s_1) \mathrm{d}\hat{\Psi}(M_\mathrm{L\to s}\mid s_2)
\prod_{i=1}^l \mathrm{d}\hat{\Phi}(M_\mathrm{f\to v}^{(1,i)}\mid x_1)
\prod_{i=1}^l \mathrm{d}\hat{\Phi}(M_\mathrm{f\to v}^{(2,i)}\mid x_2)\nonumber\\
&\quad\cdot\log\left(\sum_{(x_1,s_1,x_2,s_2)\in\mathcal{X}\times\mathcal{S}\times\mathcal{X}\times\mathcal{S}}W(y\mid x_2,s_2)V(s_1\mid y,x_2,s_2) 
M_\mathrm{R\to v}(x_1,s_1) \left(\prod_{i=1}^lM_\mathrm{f\to v}^{(1,i)}(x_1)\right)
M_\mathrm{L\to s}(s_2) \left(\prod_{i=1}^lM_\mathrm{f\to v}^{(2,i)}(x_2)\right)
\right)\nonumber\\
&\quad-
\sum_{(x,s)\in(\mathcal{X}\times\mathcal{S})}
m_\mathrm{L\to s}(s)\frac{m_\mathrm{f\to v}(x)^l}{\sum_{x\in\mathcal{X}}m_\mathrm{f\to v}(x)^l}\nonumber\\
&\quad\cdot\int\mathrm{d}\Psi(M_\mathrm{R\to v}\mid x,s) \mathrm{d}\hat{\Psi}(M_\mathrm{L\to s}\mid s)
\prod_{i=1}^l \mathrm{d}\hat{\Phi}(M_\mathrm{f\to v}^{(i)}\mid x)
\log\left(\sum_{(x,s)\in\mathcal{X}\times\mathcal{S}} M_\mathrm{R\to v}(x,s)M_\mathrm{L\to s}(s)\prod_{i=1}^l M_\mathrm{f\to v}^{(i)}(x)\right)\nonumber\\
&\quad +\frac{l}{r} \sum_{x\in\mathcal{X}^r}\frac{f(\bm{x}) \prod_{i=1}^rm_\mathrm{v\to f}(x_i)}
{\sum_{\bm{x}\in\mathcal{X}^r} f(\bm{x}) \prod_{i=1}^rm_\mathrm{v\to f}(x_i)}
\int\prod_{i=1}^r\mathrm{d}\Phi(M_\mathrm{v\to f}^{(i)}\mid x_i)
\log \left(\sum_{\bm{x}\in\mathcal{X}^r} f(\bm{x}) \prod_{i=1}^r M_\mathrm{v\to f}^{(i)}(x_i)\right)\nonumber\\
&\quad - l\sum_{x\in\mathcal{X}}\frac{m_\mathrm{f\to v}(x)m_\mathrm{v\to f}(x)}{\sum_{x\in\mathcal{X}}m_\mathrm{f\to v}(x)m_\mathrm{v\to f}(x)}
\int\mathrm{d}\Phi(M_\mathrm{v\to f}\mid x)\mathrm{d}\hat{\Phi}(M_\mathrm{f\to v}\mid x)
\log \left(\sum_{x\in\mathcal{X}} M_\mathrm{f\to v}(x)M_\mathrm{v\to f}(x)\right)\Bigg\}\label{eq:RS_mem}
\end{align}
\setcounter{equation}{8}
\end{figure*}
By using the variational method~\cite{yedidia2005constructing}, \cite{mori2011connection}, we obtain the following lemma.
\begin{lemma}\label{lem:mem-nth}
For $n\in\{0,1,2,\dotsc\}$,
\begin{align*}
&\lim_{N\to\infty}\frac1N\log\mathbb{E}[Z^n]
=\max_{(m_\mathrm{v\to f}, m_\mathrm{v\to f}, m_\mathrm{R\to v}, m_\mathrm{L\to s})\in\mathcal{R}}\bigg\{\\
&\quad \log Z_\mathrm{w} - \log Z_\mathrm{v} + \frac{l}{r}\log Z_\mathrm{f} - l\log Z_\mathrm{e}\bigg\}
-R
\end{align*}
where $\mathcal{R}$ denotes the set of saddle points of the function for which the maximum is taken, and where
\begin{align*}
Z_\mathrm{w} &:=
\sum_{\bm{x}_1,\bm{s}_1,\bm{x}_2,\bm{s}_2}
\left(\sum_{y\in\mathcal{Y}}\prod_{k=0}^nW(y\mid \bm{x}_2^{(k)},\bm{s}_2^{(k)})V(\bm{s}_1^{(k)}\mid y,\bm{x}_2^{(k)},\bm{s}_2^{(k)})\right)\\
&\quad\cdot m_\mathrm{R\to v}(\bm{x}_1,\bm{s}_1) m_\mathrm{f\to v}(\bm{x}_1)^l
m_\mathrm{L\to s}(\bm{s}_2) m_\mathrm{f\to v}(\bm{x}_2)^l \\
Z_\mathrm{v} &:=
\sum_{\bm{x},\bm{s}}  m_\mathrm{R\to v}(\bm{x},\bm{s})m_\mathrm{L\to s}(\bm{s})m_\mathrm{f\to v}(\bm{x})^l\\
Z_\mathrm{f} &:=
\sum_{\bm{x}\in(\mathcal{X}^{n+1})^r} \prod_{k=0}^n f(\bm{x}^{(k)})\prod_{i=1}^r m_\mathrm{v\to f}(\bm{x}_i)\\
Z_\mathrm{e} &:= \sum_{\bm{x}} m_\mathrm{f\to v}(\bm{x}) m_\mathrm{v\to f}(\bm{x}).
\end{align*}
The saddle point equations are
\begin{align*}
m_\mathrm{v\to f}(\bm{x}) &\propto \left(\sum_{\bm{s}} m_\mathrm{R \to v}(\bm{x},\bm{s})m_\mathrm{L\to s}(\bm{s})\right)m_\mathrm{f\to v}(\bm{x})^{l-1}\\
m_\mathrm{f\to v}(\bm{x})&\propto \sum_{i=1}^r\sum_{\substack{\bm{z}\in(\mathcal{X}^{n+1})^r\\\bm{z}_i=\bm{x}}}\prod_{k=0}^nf(\bm{z}^{(k)})
\prod_{j=1, j\ne i}^r m_\mathrm{v\to f}(\bm{z}_j)\\
m_\mathrm{R\to v}(\bm{x},\bm{s}) &\propto \sum_{\bm{x}_1, \bm{s}_1}
\left(\sum_{y\in\mathcal{Y}}\prod_{k=0}^nW(y\mid \bm{x}^{(k)},\bm{s}^{(k)})V(\bm{s}_1^{(k)}\mid y,\bm{x}^{(k)},\bm{s}^{(k)})\right)\\
&\quad\cdot m_\mathrm{R\to v}(\bm{x}_1,\bm{s}_1) m_\mathrm{f\to v}(\bm{x}_1)^l\\
m_\mathrm{L\to s}(\bm{s}) &\propto \sum_{\bm{x}_2, \bm{s}_2}
\left(\sum_{y\in\mathcal{Y}}\prod_{k=0}^nW(y\mid \bm{x}_2^{(k)},\bm{s}_2^{(k)})V(\bm{s}^{(k)}\mid y,\bm{x}_2^{(k)},\bm{s}_2^{(k)})\right)\\
&\quad\cdot m_\mathrm{L\to s}(\bm{s}_2) m_\mathrm{f\to v}(\bm{x}_2)^l.
\end{align*}
Here, the domains of sums $\bm{x}, \bm{x}_1, \bm{x}_2 \in \mathcal{X}^{n+1}$, $\bm{s}, \bm{s}_1, \bm{s}_2 \in\mathcal{S}^{n+1}$ are omitted.
\end{lemma}
Although a proof of this lemma is omitted for lack of space, the derivation is similar to the memoryless case in~\cite{mori2011connection}.

\begin{lemma}
For $n=0$, $m_\mathrm{R\to v}(x,s)$ is uniform.
Hence, the stationary condition is given by\/~\eqref{eq:saddle} and
\begin{align}
m_\mathrm{L\to s}(s) &\propto \sum_{x_2, s_2}
\left(\sum_{y\in\mathcal{Y}}W(y\mid x_2,s_2)V(s\mid y,x_2,s_2)\right)\nonumber\\
&\quad\cdot m_\mathrm{L\to s}(s_2) m_\mathrm{f\to v}(x_2)^l.
\label{eq:saddleL}
\end{align}
\end{lemma}
\begin{proof}
By considering~\eqref{eq:type},
one can confirm that $X_1$ and $S_1$ (also $X_2$ and $S_2$) should be independent for $n=0$.
As a consequence, it turns out that $m_\mathrm{R\to v}(x, s)$ is also an independent measure
(This part relates to the proof of Lemma~\ref{lem:mem-nth}. Hence, the proof is omitted for lack of space).
Then, $\sum_{x\in\mathcal{X}}m_\mathrm{R\to v}(x,s)$ should be uniform,
which is the unique right eigenvector corresponding to the eigenvalue 1 of the stochastic matrix of irreducible Markov chain $Q_0$.
Then, one straightforwardly obtains that $\sum_{s\in\mathcal{S}}m_\mathrm{R\to v}(x,s)$ is also uniform.
\end{proof}
For finite-state Markov channels,~\eqref{eq:saddleL} implies that
$m_\mathrm{L\to s}(s)$ is the stationary distribution of the Markov chain $V(s_1\mid s_2)$.

\begin{figure*}[!b]
\hrulefill
\setcounter{equation}{12}
\begin{equation}\label{eq:saddle_memory}
\begin{split}
m_\mathrm{f\to v}(x)\hat{\Phi}(M_\mathrm{f\to v}\mid x)
&= \frac1r \sum_{k=1}^r 
\sum_{\bm{x}\in\mathcal{X}^r, x_k=x}
\frac{f(\bm{x}) \prod_{i=1,i\ne k}^r m_\mathrm{v\to f}(x_i)}{\sum_{\bm{x}\in\mathcal{X}^r} f(\bm{x}) \prod_{i=1,i\ne k}^rm_\mathrm{v\to f}(x_i)}
\int\prod_{i=1, i\ne k}^r \mathrm{d}\Phi(M_\mathrm{v\to f}^{(i)}\mid x_i)\\
&\quad\cdot \delta\left(M_\mathrm{f\to v}, \frac{\sum_{\bm{x}\in\mathcal{X}^r\setminus x_k}f(\bm{x})\prod_{i\ne k}M_\mathrm{v\to f}(x_i)}
{\sum_{\bm{x}\in\mathcal{X}^r}f(\bm{x})\prod_{i\ne k}M_\mathrm{v\to f}(x_i)}\right)\\
\Phi(M_\mathrm{v\to f}\mid x)
 &= \sum_{s\in\mathcal{S}} m_\mathrm{L\to s}(s)
\int \mathrm{d}\Psi(M_\mathrm{R\to v}\mid x,s)\mathrm{d}\hat{\Psi}(M_\mathrm{L\to s}\mid s)
\prod_{i=1}^{l-1}\mathrm{d}\hat{\Phi}(M_\mathrm{f\to v}^{(i)}\mid x)\\
&\quad\cdot\delta\left(M_\mathrm{v\to f}, \frac{\left(\sum_{s\in\mathcal{S}}M_\mathrm{R\to v}(x,s)M_\mathrm{L\to s}(s)\right)\prod_{i=1}^{l-1}M_\mathrm{f\to v}^{(i)}(x)}
{\sum_{(x,s)\in\mathcal{X}\times\mathcal{S}} M_\mathrm{R\to v}(x,s)M_\mathrm{L\to s}(s)\prod_{i=1}^{l-1}M_\mathrm{f\to v}^{(i)}(x)}\right)\\
\Psi(M_\mathrm{R\to v}\mid x,s) &=
\sum_{(y,x_1,s_1)\in\mathcal{Y}\times\mathcal{X}\times\mathcal{S}}
\frac{m_\mathrm{f\to v}(x_1)^l}{\sum_{x\in\mathcal{X}}m_\mathrm{f\to v}(x)^l}
W(y\mid x,s)V(s_1\mid y,x,s) 
 \int\mathrm{d}\Psi(M_\mathrm{R\to v}'\mid x_1,s_1)\prod_{i=1}^l\mathrm{d}\hat{\Phi}(M_\mathrm{f\to v}^{(i)}\mid x_1)\\
&\quad \cdot \delta\left(M_\mathrm{R\to v},
\frac{\sum_{(x_1,s_1)\in\mathcal{X}\times\mathcal{S}} M_\mathrm{R\to v}'(x_1,s_1)\left(\prod_{i=1}^lM_\mathrm{f\to v}^{(i)}(x_1)\right) W(y\mid x,s)V(s_1\mid y,x,s)}
{\sum_{(x_1,s_1,x_2,s_2)\in\mathcal{X}\times\mathcal{S}\times\mathcal{X}\times\mathcal{S}}
 M_\mathrm{R\to v}'(x_1,s_1)\left(\prod_{i=1}^lM_\mathrm{f\to v}^{(i)}(x_1)\right) W(y\mid x_2,s_2)V(s_1\mid y,x_2,s_2)}
\right)\\
m_\mathrm{L\to s}(s)\hat{\Psi}(M_\mathrm{L\to s}\mid s) &=
\sum_{(x_2,s_2)\in\mathcal{X}\times\mathcal{S}} m_\mathrm{L\to s}(s_2) \frac{m_\mathrm{f\to v}(x_2)^l}{\sum_{x_2\in\mathcal{X}}m_\mathrm{f\to v}(x_2)^l}
\sum_{y\in\mathcal{Y}} W(y\mid x_2, s_2) V(s\mid y,x_2,s_2)
\int\mathrm{d}\hat{\Psi}(M_\mathrm{L\to s}'\mid s_2)\\
&\quad\cdot \prod_{i=1}^l\mathrm{d}\hat{\Phi}(M_\mathrm{f\to v}^{(i)}\mid x_2)
 \delta\left(M_\mathrm{L\to s}, 
\frac{\sum_{(x_2,s_2)\in\mathcal{X}\times\mathcal{S}} M_\mathrm{L\to s}'(s_2)\left(\prod_{i=1}^lM_\mathrm{f\to v}^{(i)}(x_2)\right) W(y\mid x_2,s_2)V(s\mid y,x_2,s_2)}
{\sum_{(x_2,s_2)\in\mathcal{X}\times\mathcal{S}} M_\mathrm{L\to s}'(s_2)\left(\prod_{i=1}^lM_\mathrm{f\to v}^{(i)}(x_2)\right) W(y\mid x_2,s_2)}
 \right).
\end{split}
\end{equation}
\setcounter{equation}{10}
\end{figure*}

The RS assumption for this problem is
\begin{align*}
m_\mathrm{v\to f}(\bm{x}) &= m_\mathrm{v\to f}(x_0)\int \prod_{i=1}^nM_\mathrm{v\to f}(x_i)\mathrm{d}\Phi(M_\mathrm{v\to f}\mid x_0)\\
m_\mathrm{f\to v}(\bm{x}) &= m_\mathrm{f\to v}(x_0)\int \prod_{i=1}^nM_\mathrm{f\to v}(x_i)\mathrm{d}\hat{\Phi}(M_\mathrm{f\to v}\mid x_0)\\
m_\mathrm{R\to v}(\bm{x},\bm{s}) &= m_\mathrm{R\to v}(x_0,s_0)\int \prod_{i=1}^nM_\mathrm{R\to v}(x_i,s_i)\\
&\quad\cdot\mathrm{d}\Psi(M_\mathrm{R\to v}\mid x_0,s_0)\\
m_\mathrm{L\to s}(\bm{s}) &= m_\mathrm{L\to s}(s_0)\int \prod_{i=1}^nM_\mathrm{L\to s}(s_i)\mathrm{d}\hat{\Psi}(M_\mathrm{L\to s}\mid s_0).
\end{align*}
Similarly to the previous subsection, we obtain the RS solution
as follows, which is the main result of this paper.
\begin{theorem}\label{thm:mem}
The RS solution is given by\/~\eqref{eq:RS_mem}.
The saddle point equations are\/~\eqref{eq:saddle}, \eqref{eq:saddleL} and \eqref{eq:saddle_memory}.
\end{theorem}
Although the saddle point equations in~\cite{neri2009statistical} are obtained by simplifying Theorem~\ref{thm:mem}
for finite-state Markov channels,
they are omitted in this paper for lack of space.
Similarly to the memoryless case, we have to deal with 1RSB effect generally.
Unfortunately, the frozen solution does not generally exist.
The necessary and sufficient condition of existence of the frozen solution is that
for any $y\in\mathcal{Y}$, when variables $x_2,s_2$ are fixed,
at most one value $s_1$ exists such that
\begin{equation} \label{eq:frozen}
W(y\mid x_2,s_2)V(s_1\mid y,x_2,s_2)>0
\end{equation}
and when $s_1$ is fixed at most one pair of values $(x_2,s_2)$ satisfying~\eqref{eq:frozen} exists.
This condition is called \textit{hard constraint} in~\cite{martin2005random}.
When a state depends only on the previous state or input, we can obtain a simpler expression of RS solution than Theorem~\ref{thm:mem},
and the necessary and sufficient condition for the existence of the frozen solution becomes weaker.

\section{Numerical calculation}

While Theorem~\ref{thm:mem} gives the expected conditional entropy on general Markov channel, e.g., trapdoor channel,
generally we have to deal with densities of messages similarly to the memoryless case~\cite{RiU05/LTHC},
making implementation of the density evolution involved.
In~\cite{pfister2008joint}, \textit{generalized erasure channel\/} is defined as a channel where
density evolution can be described by one parameter.
As the simplest generalized erasure channel, the dicode erasure channel (DEC) is introduced.
For $\mathcal{X}=\mathcal{S}=\{0,1\}$ and $\mathcal{Y}=\{-1,0,1,*\}$,
the DEC($\epsilon$) is defined as
\begin{align*}
W(y\mid x, s) &=
\begin{cases}
1-\epsilon,& y = x-s\\
\epsilon, & y = *
\end{cases}\\
V(s'\mid y,x,s) &= 1,\hspace{2em} \text{for } s' = x
\end{align*}
for $\epsilon\in[0,1]$.
For the DEC, $m_\mathrm{L\to s}(s)=1/2$.
We modify Theorem~\ref{thm:mem} for the case where the state only depends on the previous input,
on the basis of which the RS solution for the DEC($\epsilon$) is obtained 
without dealing with densities
as
\begin{align*}
&\lim_{N\to\infty}\frac1N\mathbb{E}[\mathbf{E}[\log Z]]
=\extr_{\epsilon_\mathrm{f\to v},\epsilon_\mathrm{v\to f},\epsilon_\mathrm{R\to v},\epsilon_\mathrm{L\to s}}
\bigg\{\\
&2(\epsilon\log \epsilon + (1-\epsilon)\log (1-\epsilon))\\
&-\bigg[\epsilon_\mathrm{R\to v}+(1-\epsilon)\epsilon_\mathrm{L\to s} + 2l \epsilon_\mathrm{f\to v} - \epsilon\epsilon_\mathrm{R\to v}\epsilon_\mathrm{f\to v}^l\\
&-\epsilon_\mathrm{f\to v}^l\left(\epsilon+\frac12(1-\epsilon)\epsilon_\mathrm{R\to v}\epsilon_\mathrm{f\to v}^l\right)
\left(\epsilon+\frac12(1-\epsilon)\epsilon_\mathrm{L\to s}\right)\bigg]\log 2\\
&-\epsilon \log \epsilon-(1-\epsilon)\log (1-\epsilon)
 + \bigg[\epsilon_\mathrm{R\to v} + (1-\epsilon)\epsilon_\mathrm{L\to s} + l\epsilon_\mathrm{f\to v}\\
&- \epsilon_\mathrm{R\to v}\epsilon_\mathrm{f\to v}^l\left(\epsilon + \frac12(1-\epsilon)\epsilon_\mathrm{L\to s}\right)\bigg]\log 2\\
&-\frac{l}{r} (1- (1-\epsilon_\mathrm{v\to f})^r)\log 2
+l(1-(1-\epsilon_\mathrm{f\to v})(1-\epsilon_\mathrm{v\to f})) \log 2
\bigg\}
\end{align*}
where the saddle point equations are
\begin{equation}
\begin{split}
\epsilon_\mathrm{f\to v} &= 1-(1-\epsilon_\mathrm{v\to f})^{r-1}\\
\epsilon_\mathrm{v\to f} &= \epsilon_\mathrm{f\to v}^{l-1}\epsilon_\mathrm{R\to v}\left(\epsilon + \frac12(1-\epsilon)\epsilon_\mathrm{L\to s}\right)\\
\epsilon_\mathrm{R\to v} &= \epsilon + \frac12(1-\epsilon)\epsilon_\mathrm{R\to v}\epsilon_\mathrm{f\to v}^l\\
\epsilon_\mathrm{L\to s} &=  \epsilon_\mathrm{f\to v}^l\left(\epsilon + \frac12(1-\epsilon)\epsilon_\mathrm{L\to s}\right).
\end{split}
\label{eq:saddleDEC}
\end{equation}
This is equivalent to the density evolution of the joint iterative decoding, which we call the belief propagation (BP) decoding, on the DEC($\epsilon$)~\cite{pfister2008joint}.
For the DEC($\epsilon$), the expected conditional entropy is
\begin{align*}
\lim_{N\to \infty} \frac1N \mathbb{E}[H(X\mid Y)] &= \bigg(\lim_{N\to \infty} \frac1N \mathbb{E}[\mathbf{E}[\log Z]]\\
&\quad -\epsilon\log\epsilon-(1-\epsilon)\log(1-\epsilon)\bigg)\frac1{\log 2}\\
= -l\epsilon_\mathrm{f\to v} + \epsilon\epsilon_\mathrm{L\to s}&
+l\frac{r-1}{r}\left(1-(1-\epsilon_\mathrm{f\to v})(1-\epsilon_\mathrm{v\to f})\right)
\end{align*}
where $\epsilon_\mathrm{f\to v}$, $\epsilon_\mathrm{v\to f}$ and $\epsilon_\mathrm{L\to s}$ satisfy~\eqref{eq:saddleDEC}.

In Fig.~\ref{fig:DEC}, 
the results for $(2,4)$, $(3,6)$, $(4,8)$, $(5,10)$ and $(6,12)$ regular ensembles are shown.
When $\epsilon$ is small, there exists only a trivial saddle point, which always gives 0 expected conditional entropy.
A non-trivial saddle point appears above BP threshold.
However, similarly to the memoryless case, except for $(2,4)$-regular ensembles,
the non-trivial saddle point yields negative expected conditional entropy for some regions as in Fig.~\ref{fig:DEC}.
In the region, while BP decoding fails, MAP decoding succeeds since the trivial saddle point should be chosen for the non-negative conditional entropy.
Above the MAP threshold, the non-trivial saddle point exhibits positive expected conditional entropy.
As degrees increase, the BP threshold decreases while the MAP threshold increases.
In contrast to the memoryless case, the BP threshold of $(2,4)$-regular ensemble is higher than that of $(3,6)$-regular ensemble.
The BP threshold and the MAP threshold for $(3,6)$-regular ensemble is about 0.568\,91 and 0.638\,65, respectively.
It implies that the upper bound of the MAP threshold in~\cite{wang2008upper} is tight.
Note that the Shannon threshold of rate-half codes is $\epsilon=(1+\sqrt{17})/8\approx 0.640\,39$.
Similar figures for symmetric memoryless channels and asymmetric memoryless channels can be found in~\cite{di2004thesis} and \cite{neri2008gallager}, respectively.

The method in this paper can be generalized to other models, e.g., 
IRA/ARA LDPC codes~\cite{wang2008upper}, CDMA on a channel with memory, compressed sensing of a Markov source, etc.

\begin{figure}[t]
\psfrag{e}{$\epsilon$}
\psfrag{H(X|Y)}{$\lim_{N\to\infty}(1/N)\mathbb{E}[H(X\mid Y)]$}
\includegraphics[width=\hsize]{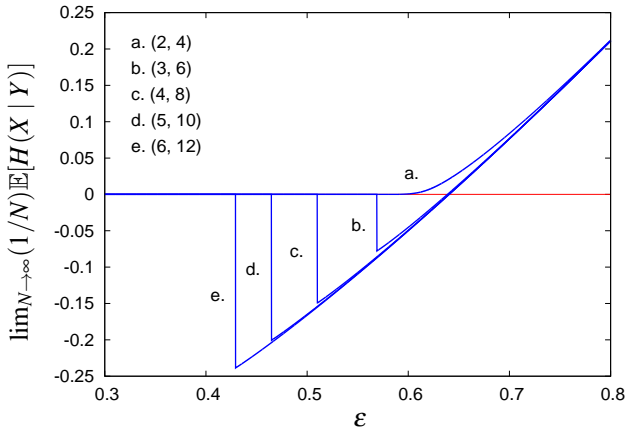}
\caption{\small The expected conditional entropy of LDPC codes on the DEC calculated by the replica method.
The results for $(2,4)$, $(3,6)$, $(4,8)$, $(5,10)$ and $(6,12)$ regular ensembles are plotted.
}
\label{fig:DEC}
\end{figure}

{\footnotesize
\bibliographystyle{IEEEtran}
\bibliography{IEEEabrv,ldpc}
}

\end{document}